\documentclass[]{llncs}
\usepackage[utf8]{inputenc}
\usepackage{amsmath}
\usepackage{etoolbox}
\usepackage{a4wide}

\usepackage{amsthm}

\usepackage{amsfonts}
\usepackage{graphicx}
\usepackage{tikz}
\usetikzlibrary{calc}
\usetikzlibrary{positioning}
\usetikzlibrary{arrows}
\usetikzlibrary{matrix}

\usepackage{makecell}
\usepackage[ruled,vlined]{algorithm2e}
\DontPrintSemicolon

\tikzstyle{robot}=[circle, draw, fill=black!80,%
                        inner sep=0pt, minimum width=4pt]

\tikzstyle{dest}=[circle, draw, fill=white,%
                        inner sep=0pt, minimum width=4pt]

\newcommand{\R}{\mathbb{R}}
\newcommand{\Z}{\mathbb{Z}}
\newcommand{\N}{\mathbb{N}}
\newcommand{\ie}{\emph{i.e.}}

\title{Stand Up Indulgent Rendezvous\thanks{This work was partially funded by the ANR project SAPPORO, ref. 2019-CE25-0005-1.}}%Rendez-vous with faulty robots}

\author{Quentin Bramas\inst{1} \and Anissa Lamani\inst{1} \and Sébastien Tixeuil\inst{2}}
\institute{ICUBE, Strasbourg University, CNRS, France \and
 Sorbonne University, CNRS, LIP6, France}
\let\oldparagraph\paragraph
\renewcommand{\paragraph}[1]{\oldparagraph{\textbf{#1}}}

\begin{document}

\maketitle

\begin{abstract}
We consider two mobile oblivious robots that evolve in a continuous Euclidean space. We require the two robots to solve the rendezvous problem (meeting in finite time at the same location, not known beforehand) despite the possibility that one of those robots crashes unpredictably. The rendezvous is stand up indulgent in the sense that when a crash occurs, the correct robot must still meet the crashed robot on its last position.

We characterize the system assumptions that enable problem solvability, and present a series of algorithms that solve the problem for the possible cases. 
%We introduce the problem of rendez-vous for two robots where 
\end{abstract}

\section{Introduction}
%AL: In progress, pas encore fini avec la partie related work\\
The study of swarm robotics in Distributed Computing has focused on the computational power of a set of autonomous robots evolving in a bidimensional Euclidean space. In this setting, a robot is modeled as a point in a two dimensional plane and has its own coordinate system and unit distance. Robots are usually assumed to be very weak : they are \emph{(i)} anonymous (they can not be distinguished), \emph{(ii)} uniform (they execute the same algorithm) and, \emph{(iii)} oblivious (they cannot remember past actions). Robots operate in cycles that comprise three phases: Look, Compute and Move. During the first phase (Look), a robot takes a snapshot to see the position of the other robots. During the second phase (Compute), a robot decides to move or stay idle. In the case in which it decides to move, it computes a target destination. In the last phase (Move), a robot moves to the computed destination (if any). Depending on how robots are activated and the synchronization level, three models have been introduced: Fully synchronous model (FFSYNC) in which robots are activated simultaneously and execute cycles synchronously. Semi synchronous model (SSYNC) in which a subset of robots are activated simultaneously. The activated robots execute a cycle synchronously. The asynchronous model (ASYNC) in which there is no global clock. The duration of each phase is finite but unbounded. That is, one robot can decide to move according to an outdated view.

Among the various problems considered under such weak assumptions there is the gathering problem which is one of the benchmarking tasks in mobile robot networks. The gathering task consists in all robots reaching a single point, not known beforehand, in finite time. The particular case of gathering two robots is called \emph{rendezvous} in the literature. In this paper, we consider the Stand Up Indulgent Rendezvous (SUIR) problem: in the case one of the two robots crashes, they still have to gather (obviously, at the position of the crashed robot); if no robot crashes, the robots are expected to gather in finite time. The SUIR problem is at least as difficult as the rendezvous problem, so classical impossibility results still apply.

\paragraph{Related works.}

A foundational result~\cite{SuzukiY99} shows that when robots operate in a fully synchronous manner, the rendezvous can be solved deterministically, while if robots are allowed to wait for a while (this is the case \emph{e.g.} in the SSYNC model), the problem becomes impossible without additional assumptions. Such additional assumptions include the robots executing a probabilistic protocol~\cite{DefagoGMP06,DefagoPCMP16} (but the rendezvous only occurs probabilistically), the robots sharing a common $x-y$ coordinate system~\cite{SuzukiY99} or an approximation of a common coordinate system~\cite{izumi12siam}, or the robots being endowed with persistent memory~\cite{0001FPSY16,HeribanDT18,SuzukiY99,Viglietta13,OkumuraWK17}. Recent work considered the minimum amount of persistent memory that is necessary to solve rendezvous~\cite{0001FPSY16,HeribanDT18,OkumuraWK17,Viglietta13}. It turns out that exactly one bit of persistent memory is necessary and sufficient~\cite{HeribanDT18} even when robots operate asynchronously and are disoriented.  

When robots can crash unexpectedly, two variants of the gathering problem can be defined~\cite{AgmonP06}: \emph{weak gathering} requires correct robots to gather, regardless of the position of crashed robots; \emph{strong gathering} requires all robots to gather at the same point. Obviously, strong gathering is only feasible if all crashed robots crash at the same location. 
Early solutions to weak gathering in SSYNC for groups of at least three robots make use of extra hypotheses: \emph{(i)} starting from a distinct configuration (that is, a configurations where at most one robot occupies a particular position), at most one robot may crash~\cite{AgmonP06}, \emph{(ii)} robots are activated one at a time~\cite{DefagoGMP06}, \emph{(iii)} robots may exhibit probabilistic behavior~\cite{DefagoPCMP16}, \emph{(iv)} robots share a common chirality (that is, the same notion of handedness)~\cite{bouzid13icdcs}, \emph{(v)} robots agree on a common direction~\cite{BCM15}. It turns out that these hypotheses are \emph{not} necessary to solve deterministic weak gathering in SSYNC, when up to $n-1$ robots may crash~\cite{bramas15wait}. %\textbf{QB:par contre, c'est sous l'hypothèse que la config initial n'est pas bivalente} \textbf{ST oui, mais sinon c'est impossible en déterministe même sans crash.} \textbf{QB: Oui mais du coup les hypothèses précédentes sont bien nécessaire dans ce cas. du coup il faudrait ptetre dire par exemple, "si on commence d'une configuration distincte ces hypothese ne sont pas nécessaire" ça évite de parler de config bivalente.}

The case of strong gathering mostly yielded impossibility results: with at most a single crash, strong gathering $n\geq3$ robots deterministically in SSYNC is impossible even if robots are executed one at a time, and probabilistic strong gathering $n\geq3$ robots is impossible with a fair scheduler~\cite{DefagoGMP06,DefagoPCMP16}. However, probabilistic strong gathering $n\geq3$ robots becomes possible in SSYNC if the relative speed of the robots is upper bounded by a constant~\cite{DefagoGMP06,DefagoPCMP16}. 
%\textbf{QB: je viens de vérifier, l'article de 2016 utilise des system de coordonnée non self-consistance donc leur preuve d'impossibilités sont moins générales, meme s'ils disent qu'ils généralisent l'impossibilité de Principe (qui lui est vrai avec des system self-consistant)} \textbf{ST: Il faut vraiment écrire ce papier sur la self-consistency... ;-)} \textbf{QB: en fait c'est moins pire que ce que je pensais car j'ai l'impression que finalement pas mal d'articles considère des system self-consistant. Il faudrait vraiment relire toute la littérature et vérifier qui fait quoi ^^}

For the special case of two robots, the strong gathering problem boils down to stand up indulgent rendezvous (SUIR), as presented above. Only few results are known:
\begin{enumerate}
\item The algorithm "with probability $\frac{1}{2}$, go to the other robot position" is a probabilistic solution to SUIR in SSYNC~\cite{DefagoGMP06,DefagoPCMP16},
\item The algorithm "go to the other robot position" is a deterministic solution to SUIR in SSYNC when exactly one robot is activated at any time~\cite{DefagoGMP06,DefagoPCMP16}.
\end{enumerate}
%\textbf{QB: tu veux dire dans les papiers précédents? on pourrait dire "In previous work"?} \textbf{ST: "to this paper" = "jusqu'à ce papier, qui enfin répond à la question".} \textbf{QB: ok j'ai lu en mode réflex "In this paper..."}
To this paper, it is unknown whether additional assumptions (\emph{e.g.} a common coordinate system in SSYNC, or FSYNC scheduling) enable deterministic SUIR solvability.

\paragraph{Our contribution.}

In this paper, we consider the SUIR problem and concentrate at characterizing its deterministic solvability. When robots share a common $x-y$ coordinate system, rendezvous is deterministically solvable in SSYNC~\cite{SuzukiY99}: the two robots simply meet at the position of the Northernmost, Easternmost position. Our main impossibility result shows that SUIR cannot be solved deterministically in this setting. Furthermore, it remains impossible if robots have \emph{both} infinite persistent memory (this is a stronger assumption than the classical luminous robot model that permits to solve classical rendezvous in ASYNC~\cite{0001FPSY16,OkumuraWK17,HeribanDT18,Viglietta13}) and a common $x-y$ coordinate system.
This motivates our focus on the FSYNC setting, where both robots always operate synchronously. Our main positive result is that SUIR is deterministically solvable in FSYNC by oblivious disoriented robots. Our approach is constructive: we first present a simple algorithm for the case the robots share a common coordinate system, and then a more involved solution for the case of disoriented robots.

%\textbf{ST: à voir si on laisse.}
An interesting byproduct of our work is an oblivious deterministic rendezvous protocol (so, assuming no faults) for the case where robots only agree on a single axis. This complements previous results where robots agree on both axes~\cite{SuzukiY99} or approximately agree on both axes~\cite{izumi12siam}.

A summary of our results is presented in the following table.

%As having common $x-y$ coordinate system is a very strong assumption, 

\begin{center}
\begin{tabular}{|c|c|c|}\hline
          & Rendezvous & SUIR \\\hline
    SSYNC & \makecell{Impossible~\cite{SuzukiY99}} & \makecell{Impossible (Theorem \ref{thm:SSYNC Impossibility with lights})} \\
    oblivious, disoriented & & \\\hline    
    SSYNC & \makecell{Possible (Algorithm~\ref{algo:one common axis})} & \makecell{Impossible (Theorem \ref{thm:SSYNC Impossibility with lights})} \\
    oblivious, common $x$ axis & & \\\hline
    SSYNC & \makecell{Possible~\cite{FlocchiniPSW05,SuzukiY99}} & \makecell{Impossible (Theorem \ref{thm:SSYNC Impossibility with lights})} \\
    oblivious, common $x-y$ axes & & \\\hline
        SSYNC & \makecell{Possible~\cite{0001FPSY16,HeribanDT18,OkumuraWK17,SuzukiY99,Viglietta13}} & \makecell{Impossible (Theorem \ref{thm:SSYNC Impossibility with lights})} \\
        luminous, disoriented & & \\\hline
    
    FSYNC & Possible~\cite{balabonski19tcs,CohenP05,SuzukiY99}    & \makecell{Possible (Algorithm~\ref{algo:disoriented FSYNC})}\\
    oblivious, disoriented & & \\\hline
    %FSYNC, one full axis & Possible \cite{?}    & \makecell{Possible (Algorithm~\ref{algo:one common axis})}\\\hline
\end{tabular}
\end{center}

\section{Model}

We consider two robots, evolving in a Euclidean two-dimensional space. Robots are anonymous and oblivious. The time is discrete, and at each time instant, called round, a subset of the robots is activated and executes a Look-Compute-Move cycle. Each activated robot first Looks at its surroundings to retrieve the position of the other robot in its ego-centered coordinate system. Then, it Computes a target destination, based only on the current position of the other robot. Finally, it moves towards the destination following a straight path.

If the movements are \emph{rigid}, each robot always reaches its destination before the next Look-Compute-Move cycle. Otherwise, movements are \emph{non-rigid}, and an adversary can stop the robot anywhere along the path to its destination, but only after the robot traveled at least a fixed positive distance $\delta$. The value of $\delta$ is not known by the robots, and can be arbitrary small, but it does not change during the execution.

In the fully-synchronous model (FSYNC), all correct robots are activated at each round. In the \emph{Semi-synchronous} model (SSYNC), only a non-empty subset of the correct robots may be activated at each round. In this case, we consider only \emph{fair} schedules \ie, schedules where each correct robot is activated infinitely often.

\paragraph{Configurations And Local Views.}
We consider different settings that impact how a robot retrieves the position of the other robot.
Robots might agree on one or both axes of their ego-centered coordinate systems. In other words, robots may have a common North (and possibly a common East direction). They may also have different unit distance. 

For the analysis, we assume a global coordinate system $Z$ that is not accessible to the robots. A \emph{configuration} is a set $\{r_1, r_2\}$ containing the positions of both robots in $Z$. Notice that $r_i$, $i=1,2$, denotes at the same time a robot and its position in $\R^2$ in the coordinate system $Z$.

\newcommand{\IT}{\mathcal{T}}
To model the agreement of the robots about their coordinate system, we define the set $\IT$ of \emph{indistinguishable transformations}, that modify how robots see the current configuration. If robots agree on both axes and on the unit distance, then $\IT$ only contains the identity. If robots do not agree on the unit distance, then $\IT$ contains all the (uniform) scaling transformations. If robots do not agree on the $x$-axis, then $\IT$ also contains the reflection along the $y$-axis. If robots does not agree on any axis, then $\IT$ also contains all the rotations. Finally, $\IT$ is closed by composition.
%\textbf{ST: j'ai ajouté "also", car je suppose que les transformations possibles s'accumulent quand on passe de 1-axis à 2-axes.}
%QB: Oui, du coup j'ai meme ajouté un also avant car les transformations s'accumulent quand on ajoute les réflexions.

We say robots are \emph{disoriented} if robots do not agree on any axis, nor on a common unit distance \ie, $\IT$ contains the rotations, scaling, reflection, and their compositions.

In a configuration $\{r_1, r_2\}$ the \emph{local view} $V_1$ of robot $r_1$ is obtained by translating the global configuration by $-r_1$ (so that $r_1$ is seen at position $(0,0)$ and $r_2$ is at position $r_2 - r_1$) from which we apply the transformation function $h_{1}$ of $r_1$. Formally, $V_1 = \{(0,0), h_{1}(r_2 - r_1)\}$. Similarly the local view of $r_2$ is $V_2 = \{ (0,0),h_{2}(r_1 - r_2)\}$, where $h_{2} \in \IT$ corresponds to the transformation function of robot $r_2$. 
%\textbf{ST: je ne comprends pas la notation $r_2-r_1$.} 
Notice that the transformation function of a robot $r$ is chosen by an adversary but it does not change over time. % (we say robots are \emph{self-consistent} as opposed to \emph{self-inconsistent}).
%When robots agree on both axis, $h_r$ is replaced by the multiplication by $\lambda_r\in \R$, where $\lambda_r$ represent the unit distance of robot $r$.
  
A configuration $C$ is said to be \emph{distinct} if $|C| = 2$.

\paragraph{Configurations And Local Views in One Dimension.}
The evolving space of the robots can be naturally restricted to a one-dimensional space \ie, robots that evolve on a line. In this case, robots in a configuration $C$ correspond to points in $\R$ instead of $\R^2$. Similarly, the set of transformation functions $\IT$ contains scaling if robots do not agree on the unit distance, and contains the reflection (or equivalently the $\pi$-rotation) if robots do not agree on the orientation of the single axis (\ie, are disoriented). 

\paragraph{Algorithms And Executions.}
An algorithm $A$ is a function mapping local views to destinations.
The local view of a robot $r$ is centered and transformed by a function $h_r$, and when activated, algorithm $A$ outputs $r$'s destination $d$ in its local view. So to obtain the destination of a robot in the global coordinate system $Z$, one should apply the inverse transformation $h_r^{-1}$ \ie, the global destination is $r+h_r^{-1}(d)$.

When a non-empty subset of robots $S$ executes an algorithm $A$ in a given configuration $C$, the obtained configuration $C'$ is the smallest set satisfying\footnote{This definition works when $|C| = 2$ but can be easily generalized to larger configurations}:
\begin{align*}
    \forall r\in C\setminus S &\Rightarrow r\in C'\\
    \forall r\in C\cap S      &\Rightarrow  r+h_r^{-1}\left(A\left(\{(0,0), h_r(r'-r)\}\right)\right)\in C', \text{with $r'\in C\setminus\{r\}$}.
\end{align*}
In this case, we write $C \overset{A}{\rightarrow} C'$, and say $C'$ is obtained from $C$ by applying $A$.% ($A$ may be omitted if clear from the context).
We say a robot \emph{crashes} at time $t$ if it is not activated at time $t$, and never activated after time $t$ \ie, a crashed robot stops executing its algorithm and remains at the same position.

\emph{An execution} of algorithm $A$ is an infinite sequence of configurations $C_0, C_1, \ldots$ such that $C_i\overset{A}{\rightarrow} C_{i+1}$ for all $i\in \N$.
We say an execution contains one crashed robot, if one robot crashes at some round $t$.

\paragraph{The Stand Up Indulgent Rendezvous Problem.}

An algorithm solves the Stand Up Indulgent Rendezvous (SUIR) problem if, for any initial configuration $C_0$ and for any execution $C_0, C_1, \ldots$ with up to one crashed robot, there exists a round $t$ and a point $p$ such that $C_{t'} = \{p\}$ for all $t'\geq t$.

Informally, if one robot crashes, the correct robot goes to the crashed robot; if no robot crashes, both robots gather in a finite number of rounds. 

Since we allow arbitrary initial configuration and the robots are oblivious, we can consider without loss of generality that the crash, if any, occurs at the start of the execution.

\section{Impossibility Results}

In this section we prove that the SUIR problem is not solvable in SSYNC, even if robots share a full coordinate system (the transformation function is the identity), and have access to infinite persistent memory that is readable by the other robot. In the literature~\cite{0001FPSY16}, the persistent memory aspect is called the Full-light model with an infinite number of colors. We now prove that having such capabilities does not help solving the problem. The next lemma is very simple but is a key point to prove our main impossibility result.

%\begin{lemma}\label{lem:both robots move}
%In any distinct configuration, the SR algorithm must command both robots to move by some %distance $d>0$.
%\end{lemma} 
%\begin{proof}
%If there exists a configuration $C$ such that one robot $r$ is \emph{not} commanded to move, then, if the other robot crashes, the configuration is invariant and the two robots never merge.
%\end{proof}

%\movetoappendice{\ref{lem:execution contain move to other}}
{
\begin{lemma}\label{lem:execution contain move to other}
Consider the SSYNC model, with rigid movements, robots endowed with full-lights with infinitely many colors, and a common coordinate system. Assuming algorithm $A$ solves the SUIR problem, then, in every execution suffix starting from a distinct configuration where only robot $r$ is activated (\emph{e.g.} because the other robot has crashed), there must exist a configuration where algorithm $A$ commands that $r$ moves to the other robot's position.%\textbf{ST: ajouté pourquoi on active que $r$.} QB: Ok
\end{lemma}
}
{
\begin{proof}
Any move of robot $r$ that does not go to the other robot location does not yield gathering. If this repeats infinitely, no rendezvous is achieved.
\end{proof}
}

\begin{theorem}\label{thm:SSYNC Impossibility with lights}
The SUIR problem is not solvable in SSYNC, even with rigid movements, robots endowed with full-lights with infinitely many colors, and sharing a common coordinate system.
\end{theorem}
\begin{proof}
%\textbf{ST: reformulé la preuve.} QB: OK
Assume for the purpose of contradiction that such an algorithm exists. Let $r$ be one of the robots, and $r'$ be the other robot. We construct a fair infinite execution where rendezvous is never achieved. At some round $t$, we either activate only $r$, only $r'$, or both, depending on what the (deterministic) output of the algorithm in the current configuration is:
\begin{itemize}
    \item If $r$ is dictated to stay idle: activate only $r$
    \item If $r$ is dictated to move to $p\neq r'$: activate only $r$
    \item If $r$ is dictated to move to $r'$, and $r'$ is dictated to move: activate both robots.
    \item Otherwise ($r'$ is dictated to stay idle): activate only $r'$
\end{itemize}

We now show that the execution is fair. 
Suppose for the purpose of contradiction that the execution is unfair, so there exists a round $t$ after which only $r$ is executed, or only $r'$ is executed. In the first case, it implies there exists an execution suffix where $r$ is never dictated to move to the other robot, which contradicts Lemma~\ref{lem:execution contain move to other}. Now, if only $r'$ is activated after some round $t$, then there exists a suffix where $r'$ is always dictated to stay idle, which also contradicts Lemma~\ref{lem:execution contain move to other}.

The schedule we choose guarantees the following.
When only $r$ is activated, rendezvous is not achieved as $r$ is not moving to $r'$.
When only $r'$ is activated, rendezvous is not achieved as $r'$ is idle. If both robots are activated rendezvous is not achieved as $r$ is moving to $r'$ while $r'$ is moving. 

Overall, there exists an infinite fair execution where robots never meet, a contradiction with the initial assumption that the algorithm solves SUIR.
\end{proof}

\section{Reduction To One-dimensional Space}

In this section, we prove that having an algorithm solving the SUIR problem in a one-dimensional space implies the existence of an algorithm solving the same problem in a two-dimensional space. This theorem is important as our algorithms are defined on the one-dimensional space. However, since we do not prove the converse, the impossibility result we present in the previous Section works in the most general settings, assuming a two-dimensional space. Indeed, we present after the theorem, an example of a problem (the fault-free rendezvous with one common full axis) that is solvable in a two dimensional space, but that cannot be reduced to the one-dimensional space. Despite the results being intuitive, the formal proof is not trivial.

\begin{theorem}\label{thm:1D to 2D}
Suppose there exists an algorithm solving the SUIR problem where robots are restricted to a one-dimensional space. Then, there exists an algorithm solving the SUIR problem in a two-dimensional space. 
\end{theorem}

\begin{proof}
Let $A_1$ be an algorithm solving the SUIR problem where robots are restricted to a one-dimensional space. We provide a constructive proof of the theorem by giving a new algorithm $A_2$ executed by robots in the two-dimensional space.

\newcommand{\vfunc}{\textbf{v}}

First, for a configuration $C = \{r_{min}, r_{max}\}$, with $r_{min} < r_{max}$ (using the lexicographical order on their coordinates), we define the function $\vfunc$ as follows:
\[
\vfunc(\{r_{min}, r_{max}\}) = \frac{r_{\max} - r_{\min}}{\lVert r_{\max} - r_{\min} \rVert}
\]

Let $r_1$ and $r_2$ denote the two robots, having transformation functions $h_1$ and $h_2$, respectively.
Let $V_i = \{(0,0), h_i(r_j - r_i) \}$, $i=1,2$, $j=3-i$, be the local view of robot $r_i$. Each robot $r_i$ can compute its own orientation vector $v_i = \vfunc(V_i)$ of the line joining the two robots. Notice that, if robots remains on the same line, then $v_i$ remains invariant during the whole execution as long as robots do not gather.%(while robots are not yet gathered).

We define algorithm $A_2$, executed by robots in the two-dimensional space as follow. First, if the local view of a robot $r_i$ is $\{(0,0)\}$, then $A_2$ outputs $(0,0)$. 

Otherwise, $r_i$ can map its local view $V_i$ in a one-dimensional space to obtain $\overline{V_i} = \{0, b_i\}$ with $b_i$ such that $h_i(r_j - r_i) = b_i v_i$, and execute $A_1$ on $\overline{V_i}$. The obtained destination $\overline{p} = A_1(\overline{V_i})$, is then converted back to the two-dimensional space to obtain the destination $p = \overline{p} v_i$.
Doing so, the robots remain on the same line, and $v_i$ remains invariant while the robots are not gathered (when they are gathered, both algorithms stop).

Let $E = C_0, C_1, \ldots$ be an arbitrary execution of $A_2$. We want to construct from $E$ an execution $\overline{E}$ of $A_1$ such that the rendezvous is achieved in $\overline{E}$ if and only if the rendezvous is achieved in $E$.

Recall that we analyze each configuration $C_i$, $i\in\N$, using $Z$, the global coordinate system we use for the analysis. Let $v = \vfunc(C_0)$. Again, since robots remain on the same line, then $v = \vfunc(C_i)$ for any $C_i$ while robots are not yet gathered.

Let $O$ be any point of the line $L$ joining the two robots.
We define as follow the bijection $m$ mapping points of $L$ (in $Z$), to the global one-dimensional coordinate system $(O,v)$:
\[
\forall a\in\R, m(O + a v) = a
\]
We can extend the function $m$ to configurations as follows: \\$m(C) = \{m(r)\;|\; r\in C\}$.

Now, let $\overline{E}=\overline{C_0}, \overline{C_1},\ldots$ be the execution of $A_1$, in $(O,v)$, of two robots having transformation function $\overline{h}_1$ and $\overline{h}_2$ respectively, with $\overline{h}_i(a) = b$ if and only if $h_i(a v) = b v_i$.

\tolerance=1000
We now show that, if $C\overset{A_2}{\rightarrow} C'$ then $m(C)\overset{A_1}{\rightarrow} m(C')$. To do so we show that the result of executing $A_1$ on $m(C)$ coincides with $m(C')$. 
Let ${C=\{O+a_1v, O+a_2v\}}$, $i$ be an activated robot and $j$ be the other robot. On one hand, we have $m(C) = \{a_1, a_2\}$ and, by construction, the view $\overline{V_i}$ of robot $i$ in $m(C)$ is $\{0, b_i\}$ with $b_i = \overline{h}_i(a_j - a_i)$, so that the global destination of $r_i$ in $m(C)$ is then $a_i + \overline{h}_i^{-1}(\overline{p})$ (with $\overline{p} = A_1(\overline{V_i})$). 
On the other hand, the view ${V_i}$ of robot $i$ is $\{(0,0), h_i((a_j-a_i)v)\} = \{(0,0), b_iv_i\}$, so that the global destination of $r_i$ in $Z$ is $O+a_iv + h_i^{-1}(\overline{p}v_i) = O+a_iv + \overline{h}_i^{-1}(\overline{p})v $. Since $m( O+a_iv + \overline{h}_i^{-1}(\overline{p})v) = a_i + \overline{h}_i^{-1}(\overline{p})$, we obtain that $m(C)\overset{A_1}{\rightarrow} m(C')$ (assuming the same robots are activated in $C$ and in $m(C)$.

Hence, $\overline{C_i} = m(C_i)$ for all $i\in \N$.
Since $A_1$ solves the SUIR problem, there exists a point $p\in \R$ and a round $t$ such that, for all $t'\geq t$, $m(C_{t'}) = \{p\}$. This implies that $C_{t'} = \{ O + p v \}$, so that $A_2$ solves the SUIR problem.
\end{proof}

\paragraph{Rendezvous without faults with one full axis.}

Now, we show that the converse of Theorem~\ref{thm:1D to 2D} is not true in the fault-free model. This observation justifies that, for the results to be more general, we defined our model and gave our impossibility results for the two-dimensional space. 

We present an algorithm solving the (fault-free) rendezvous problem in a two-dimensional space, assuming robots agree on one full axis (that is, they agree on the direction and the orientation of the axis). Under this assumption, it is possible that the robots do not agree on the orientation of the line joining them, so that assuming the converse of Theorem~\ref{thm:1D to 2D} would imply the existence of an algorithm in the one-dimensional space with disoriented robots (which does not exists, using a similar proof as the one given in~\cite{SuzukiY99}).

The idea is that, if the configuration is symmetric (robots may have the same view), then robots move to the point that forms, with the two robots, an equilateral triangle. Since two such points exist, the robots choose the northernmost one (the robots agree on the $y$-axis, which provides a common North). Otherwise, the configuration is not symmetric, and there is a unique northernmost robot $r$. This robot does not move and the other robot moves to $r$.

\begin{algorithm}[H] 
\KwData{
  $r$ : robot executing the algorithm
 }
Let $\{(0,0), (x,y)\}$ be $r$'s local view.\\
\uIf{$y = 0$}{
    \nl\label{algo:fault-free rdv:move to triangle} move to the point $\left(x/2, \left|\frac{x\sqrt{3}}{2}\right|\right)$
}
\uElseIf{$y > 0$}{
    \nl\label{algo:fault-free rdv:move to other} move to the other robot, at $(x,y)$.
}
 \caption{Fault-free rendezvous. Robots agree on one full axis, may not have common unit distance, and movements are non-rigid}\label{algo:one common axis}
\end{algorithm}

\begin{theorem}\label{thm: proof of algo, one axis, non-rigid}
Algorithm~\ref{algo:one common axis} solves the (fault-free) rendezvous problem with non-rigid movements and robots having only one common full axis, and different unit distance.
\end{theorem}

\begin{proof}
If the configuration is not symmetric, then, after executing Line~\ref{algo:fault-free rdv:move to other}, either the rendezvous is achieved, or the moving robots remains on the same line joining the robots, so the obtained configuration is still asymmetric, and the same robot is dictated to move towards the same destination, so it reaches it in a finite number of rounds.

Consider now that the initial configuration $C$ is symmetric. 
If both robots reaches their destination, the rendezvous is completed in one round. If robots are stopped before reaching their destinations, two cases can occur. Either they are stopped after traveling different distances, or they are stopped at the same $y$-coordinate. In the former case, the obtained configuration is asymmetric and we retrieve the first case. In the latter, the configuration remains symmetric, but the distance between the two robots decreases by at least $\frac{2\delta}{\sqrt{2}}$.
%(as illustrated in Figure~\ref{fig:moving to the vertex of the equilateral triangle}). 
Hence, at each round either robots complete the rendezvous, reach an asymmetric configuration, or come closer by a fixed distance. Since the latter case cannot occurs infinitely, one of the other case occurs at least once, and the rendezvous is completed in a finite amount of rounds.
%
%\begin{figure}[h]
%    \centering
%    \input{triangle-equilateral.tex}
%    \caption{Robots moving to the third vertex of the equilateral triangle}
%    \label{fig:moving to the vertex of the equilateral triangle}
%\end{figure}
\end{proof}

%%%%%%%%%%%%
%%%%%%%%%%%%
%%%%%%%%%%%%

\section{SUIR Algorithm for FSYNC Robots with a Common Coordinate System}
%Rendez-vous when Robots Agree on Both Axis}

%\subsection{Rendez-vous on a continuous line with rigid movement and common distance unit}

Since it is impossible to solve the SUIR problem in SSYNC, we now concentrate on FSYNC. We first consider a strong model, assuming robots agree on both axis, have a common unit distance, and assuming movements are rigid, before relaxing all hypotheses in Section~\ref{sec:relax}. 
By Theorem~\ref{thm:1D to 2D}, it is sufficient to give an algorithm for a one-dimensional space. 
Figure~\ref{fig:two possible configurations} illustrates the two possible configurations: either the distance between the two robots is greater than one (the common unit distance), or at most one. In the former case both robots move to the middle. In the latter, we can dictate the right robot to move to the left one, and the left robot to move one unit distance to the right of the other robot. Recall that we can distinguish the left and the right robot on the line because we assume the robots agree on both axis in the two-dimensional space.

\begin{algorithm}[H] 
\KwData{
  $r$ : robot executing the algorithm\\
  $d$ : the distance between the two robots
}
\lIf{$d > 1$}{
    move to the middle
}
\uElse{
    \uIf{$r$ is the left robot}{
        move to the point at distance one at the right of the other robot
    }
    \uElse{
        move to the other robot
    }
}
 \caption{Rendezvous with rigid movements, and a common coordinate system}\label{algo:rigid movement and common distance}
\end{algorithm}

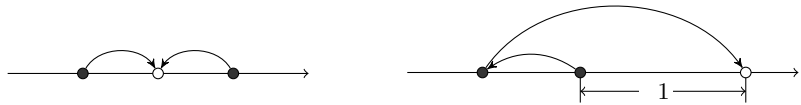
\begin{figure}[h]
    \centering
\begin{tikzpicture}
\draw[->] (-1,0) -- (3,0);
\node[robot] (a) at (0,0) {};
\node[robot] (b) at (2,0) {};
\node[dest] (m) at (1,0) {};
\node[opacity=0, text opacity=1] at (0.2,-0.6) {case $d>1$};

\path (a) edge[bend left=60,->,>=stealth'] node [left] {} (m);
\path (b) edge[bend right=60,->,>=stealth'] node [left] {} (m);
\end{tikzpicture}~~~~~~~~~~~
\begin{tikzpicture}
\draw[<->] (1.3,-0.25) -- (3.5,-0.25);
\draw[-] (1.3,0) -- (1.3,-0.4);
\draw[-] (3.5,0) -- (3.5,-0.4);

\draw[->] (-1,0) -- (4.2,0);
\node[robot] (a) at (0,0) {};
\node[robot] (b) at (1.3,0) {};
\node[dest] (r) at (3.5,0) {};
\node[opacity=0, text opacity=1] at (0.2,-0.6) {case $d\leq 1$};

\node[opacity=1, fill=white, inner sep = 1pt, text opacity=1] at (2.3,-0.25) {$1$};

\path (a) edge[bend left=55,->,>=stealth'] node [left] {} (r);
\path (b) edge[bend right=35,->,>=stealth'] node [left] {} (a);
\end{tikzpicture}
    \caption{The two possible configurations, depending the distance between the two robots}
    \label{fig:two possible configurations}
\end{figure}

\begin{theorem}
Algorithm~\ref{algo:rigid movement and common distance} solves the SUIR problem in FSYNC with rigid movement, and robots agreeing on both axes and unit distance.
\end{theorem}

\begin{proof}
%Algorithm~\ref{algo:rigid movement and common distance} is defined in the one-dimensional space but, using Theorem~\ref{thm:1D to 2D}, it induces an algorithm working in the two-dimensional space.
First we consider the case where a robot crashes.
If the left robot crashes, the right robot halves its distance with the other robot each round until its distance is at most one. Then, the right robot move to the other robot and the rendezvous is achieved.

If the right robot crashes, the left robot halves its distance with the other robot each round until its distance is at most one unit. Then, the left robot moves to the right of the other robot, at distance one. It then move to the other robot and the rendezvous is achieved.

Now assume no robot crashes.
If the configuration is such that the distance $d$ between the two robots is greater than one, then both robots move to the middle at the same time and the rendezvous is achieved.
Otherwise, if the distance $d$ is at most one, after one round the robots are at distance $1+d > 1$, so that after one more round, the rendezvous is achieved.
\end{proof}

%%%%%%%%%%%%%%%%%%%%%%%%%%%%%%%%%%%%%%%%%%%%%%%%%%%%%%%%%

\section{SUIR Algorithm for Disoriented Robots in FSYNC}
\label{sec:relax}

In this section we present Algorithm~\ref{algo:disoriented FSYNC}, which works with disoriented robots (robots do not agree on any axis, nor on the unit distance), and non-rigid movements.
The algorithm is defined on the line. Each robot sees the line oriented in some way, but robots might not agree on the orientation of the line. However, since the orientation of the line is deduced from the robot own coordinate system, it does not change over time.

\begin{algorithm}[H]
 %\KwData{
 % $r$ : robot executing the algorithm\\
  Let $d$ be the distance to the other robot 
 %}

Let $i\in\mathbb{Z}$ such that $d\in [2^{-i}, 2^{1-i})$\;
\lIf{$i \equiv 0 \mod 2$}{
    move to the middle
}
\lIf{$i \equiv 1 \mod 4$}{
    \emph{left} $\rightarrow$ move to middle ;
    \emph{right} $\rightarrow$ move to other
}
\lIf{$i \equiv 3 \mod 4$}{
    \emph{left} $\rightarrow$ move to other ;
    \emph{right} $\rightarrow$ move to middle
}
 \caption{SUIR Algorithm for disoriented robots}\label{algo:disoriented FSYNC}
\end{algorithm}

\begin{figure}[h]
    \centering
\begin{tikzpicture}
\draw[->] (-1,0) -- (3,0);
\node[robot] (a) at (0,0) {};
\node[robot] (b) at (2,0) {};
\node[dest] (m) at (1,0) {};
\node[opacity=0, text opacity=1] at (-4,0) {case $i\equiv 0\mod 4$};

\path (a) edge[bend left=60,->,>=stealth'] node [left] {} (m);
\path (b) edge[bend right=60,->,>=stealth'] node [left] {} (m);
\end{tikzpicture}

\begin{tikzpicture}
\draw[-, white] (-5,-0.2) -- (2.5,0.5); % limit
\draw[->] (-1,0) -- (3,0);
\node[robot] (a) at (0,0) {};
\node[robot] (b) at (2,0) {};
\node[dest] (m) at (1,0) {};
\node[opacity=0, text opacity=1] at (-4,0) {case $i\equiv 1\mod 4$};

\path (a) edge[bend left=30,->,>=stealth'] node [left] {} (m);
\path (b) edge[bend right=50,->,>=stealth'] node [left] {} (a);
\end{tikzpicture}

\begin{tikzpicture}
\draw[-, white] (-5,-0.2) -- (2.7,0.6); % limit
\draw[->] (-1,0) -- (3,0);
\node[robot] (a) at (0,0) {};
\node[robot] (b) at (2,0) {};
\node[dest] (m) at (1,0) {};
\node[opacity=0, text opacity=1] at (-4,0) {case $i\equiv 2\mod 4$};

\path (a) edge[bend left=60,->,>=stealth'] node [left] {} (m);
\path (b) edge[bend right=60,->,>=stealth'] node [left] {} (m);
\end{tikzpicture}

\begin{tikzpicture}
\draw[-, white] (-5,-0.2) -- (3,0.5); % limit
\draw[->] (-1,0) -- (3,0);
\node[robot] (a) at (0,0) {};
\node[robot] (b) at (2,0) {};
\node[dest] (m) at (1,0) {};
\node[opacity=0, text opacity=1] at (-4,0) {case $i\equiv 3\mod 4$};

\path (a) edge[bend left=50,->,>=stealth'] node [left] {} (b);
\path (b) edge[bend right=30,->,>=stealth'] node [left] {} (m);
\end{tikzpicture}
    \caption{The four possible configurations, depending the distance between the two robots. We have split the case $i\equiv 0 \mod 2$ into two lines to help the reader.}
    \label{fig:four possible configurations}
\end{figure}
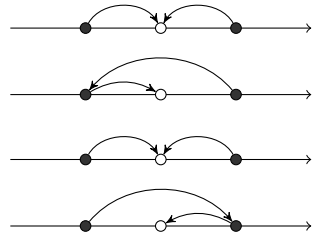

The different moves of a robot $r$ depend on whether a $r$ sees itself on the left or the right of the other robot, and on its \emph{level}. The level of robot at distance $d$ from the other robot (according its own coordinate system, hence its own unit distance), is the integer $i\in\Z$ such that $d\in [2^{-i}, 2^{1-i})$.
Figure~\ref{fig:four possible configurations} summarizes the eight possible views of a robot $r$, and the corresponding movements. Each line represents the congruence of the level of the robot modulo four, and on each line, we see the movement of the robot whether it sees itself on the right or on the left of the other robot. A given figure does not necessarily imply that both robots will actually perform the corresponding movement at the same time (since they may \emph{not} have the same view).

For instance, if a robot $r_1$ has a level $i_1$ congruent to 1 modulo 4 and sees itself on the right, while the other robot $r_2$ has a level $i_2$ congruent to 2 modulo 4, and also sees itself on the right, then $r_1$ moves to the other robot position, and $r_2$ moves to the middle. Assuming both robots reach their destinations, then the distance between them is divided by two (regardless of the coordinate system) so their levels increase by one, and they both see the other robot on the other side, so each robot now sees the other robot on its left.

Let $C$ be any configuration and $d$ is the distance (in the global coordinate system $Z$) between the two robots. Let $x$, resp. $y$, be the distance, in $Z$, traveled by the left robot, resp. the right robot. Since the robots move toward each other, after executing one round, the distance between the robot becomes $f(d,x,y) = |d - x - y|$.

\begin{lemma}\label{lem:f decreases by fixed distance}
If at least one robot is dictated to move to the middle, then we have $f(d,x,y) \leq d - \min(\delta, d/2)$.
\end{lemma}
\begin{proof}
For any fixed $d$, using the symmetry of $f$ (with respect to the second and third argument), we have $f(d,x,y) = g_d(x+y)$ with $g_d: w\mapsto |d - w|$.
We know that the distance traveled by the robots is either $0$ (if one robot crashes), or at least $\min(d/2, \delta)$, but we cannot have $x=y=0$. Also, since at least one robot moves to the middle, we have either \emph{(i)} $x \leq d/2$ and $y\leq d$, or \emph{(ii)} $x \leq d$ and $y\leq d/2$. Hence, the sum $x+y$ is in the interval $[\min(d/2, \delta), 3d/2]$. 

As a convex function, the maximum of $g_d$ is reached at the boundary of its domain
\begin{align*}
    f(d,x,y) = g_d(x+y) &\leq \max(g_d(3d/2), g_d(\min(d/2, \delta)))\\
&\leq \max(d/2, d - \min(d/2, \delta)) =  d - \min(d/2, \delta)
\end{align*}
\end{proof}

The next lemma is a direct consequence of Lemma~\ref{lem:f decreases by fixed distance}.%, but due to space constraints, its proof has been ommited. 
%moved to the appendix.
%\movetoappendice{\ref{lem:distance decreases by fixed distance}}
{
\begin{lemma}\label{lem:distance decreases by fixed distance}
From a configuration where robots are at distance $d$ (in $Z$), then, after two rounds, the distance between the robots decreases by at least $\min(\delta, d/2)$.
\end{lemma}
}{
\begin{proof}
From Lemma \ref{lem:f decreases by fixed distance}, we know that if at least one robot is dictated to move to the middle, then, after one round, the distance between the two robots decreases by at least $\min(\delta, d/2)$. 
Otherwise,  we know that both robots were dictated to move to the other robot location. Hence, the distance between the robots is either at most $d - \min(\delta, d/2)$, or is greater than $d - \min(\delta, d/2)$. In the former case, the lemma is proved. In the latter case the order of the robots changes (a left robot becomes a right robot, and vice and versa). This happens regardless of their coordinates system (maybe both robots view themselves on the right, then they both view themselves at the left). Also, the level of each robot is either the same, or is increased by one. In all cases, both robots are dictated to move to the middle.
In more details, a right robot with level $i\equiv 1\mod 4$ becomes a left robot with either the same level of level $i+1\equiv 2\mod 4$. In both cases, in the next round, the robot is dictated to move to the middle.
A left robot with level $i\equiv 3\mod 4$ becomes a right robot with either the same level of level $i+1\equiv 0\mod 4$. In both cases, in the next round, the robot is dictated to move to the middle.

So that after one more round, the distance decreases by at least $\min(\delta, d/2)$.
\end{proof}
}

\begin{lemma}\label{lem: proof of algo, rigid, crash}
Assuming rigid movement and one robot crash, Algorithm~\ref{algo:disoriented FSYNC} solves the SUIR problem.
\end{lemma}

\begin{proof}
Let $i$ be the level of the correct robot $r$. Assume the other robot crashes. Robot $r$ either sees itself on the right or on the left of the other robot.

If $r$ sees itself on the right, then depending on its level, either $r$ moves to the middle, or move to the other robot. In the former case, the level of $r$ increases by one and $r$ continues to see itself on the right. In the latter case, the rendezvous is achieved in one round. After at most three rounds, the level of $r$ is congruent to 1 modulo 4 so that after at most four rounds the rendezvous is achieved.

Similarly, if $r$ sees itself on the left, then after at most four rounds, $r$ level is congruent to 3 modulo 4 and the rendezvous is achieved.
\end{proof}

If, at round $t$, one robot sees itself on the right, and the other sees itself on the left, then they agree on the orientation of the line at time $t$. Since, for each robot, the orientation of the line does not change, then they agree on it during the whole execution (except when they are gathered, as the line is not defined in that case).

Similarly, if at some round, both robots see themselves at the right (resp. at the left), then their orientations of the line are opposite, and remain opposite during the whole execution (again, until they gather). Hence we have the following remark.

\begin{remark}\label{rem:orientation common or opposite}
Consider two disoriented robots moving on the line $L$ joining them and executing Algorithm~\ref{algo:disoriented FSYNC}. Then, either they have a common orientation of $L$ during the whole execution (while they are not gathered), or they have opposite orientations of $L$ during the whole execution (while they are not gathered).
\end{remark}

\begin{lemma}\label{lem: proof of algo, correct, common}
Assuming rigid movements, no crash, and robots having \textbf{a common orientation} of the line joining them, then, Algorithm~\ref{algo:disoriented FSYNC} solves the SUIR problem.
\end{lemma}
\begin{proof}
Since the robots have a common orientation, we know there is one robot that sees itself on the right and and one robot that sees itself on the left. Of course the robots are not aware of this, but we saw in the previous remark that a common orientation is preserved during the whole execution (while robots are not gathered).

Let $(i,j) \in \Z^2$ denote a configuration where the robot on the left is at level $i$, and the robot on the right is at level $j$, and we write $(i,j)\equiv (k,l) \mod 4$ if and only if $i\equiv k \mod 4$ and $j\equiv l \mod 4$.

To prove the lemma we want to show that for any configuration $(i,j) \in \Z^2$, the robots achieve rendezvous. 
Take an arbitrary configuration $(i,j) \in \Z^2$. We consider all 16 cases:

\begin{enumerate}
\item \label{i:com:0-0} \textbf{if} $(i,j)\equiv (0,0) \mod 4$: rendezvous is achieved in one round.
\item \label{i:com:0-1} \textbf{if} $(i,j)\equiv (0,1) \mod 4$:  we reach configuration $(j+1, i+1) \equiv (2,1) \mod 4$
\item \label{i:com:0-2} \textbf{if} $(i,j)\equiv (0,2) \mod 4$: rendezvous is achieved in one round.
\item \label{i:com:0-3} \textbf{if} $(i,j)\equiv (0,3) \mod 4$: rendezvous is achieved in one round.
\item \label{i:com:1-0} \textbf{if} $(i,j)\equiv (1,0) \mod 4$: rendezvous is achieved in one round.
\item \label{i:com:1-1} \textbf{if} $(i,j)\equiv (1,1) \mod 4$:  we reach configuration $(j+1, i+1) \equiv (2,2) \mod 4$
\item \label{i:com:1-2} \textbf{if} $(i,j)\equiv (1,2) \mod 4$: rendezvous is achieved in one round.
\item \label{i:com:1-3} \textbf{if} $(i,j)\equiv (1,3) \mod 4$: rendezvous is achieved in one round.
\item \label{i:com:2-0} \textbf{if} $(i,j)\equiv (2,0) \mod 4$: rendezvous is achieved in one round.
\item \label{i:com:2-1} \textbf{if} $(i,j)\equiv (2,1) \mod 4$:  we reach configuration $(j+1, i+1) \equiv (2,3) \mod 4$
\item \label{i:com:2-2} \textbf{if} $(i,j)\equiv (2,2) \mod 4$: rendezvous is achieved in one round.
\item \label{i:com:2-3} \textbf{if} $(i,j)\equiv (2,3) \mod 4$: rendezvous is achieved in one round.
\item \label{i:com:3-0} \textbf{if} $(i,j)\equiv (3,0) \mod 4$:  we reach configuration $(j+1, i+1) \equiv (1,0) \mod 4$
\item \label{i:com:3-1} \textbf{if} $(i,j)\equiv (3,1) \mod 4$:  we reach configuration $(j, i) \equiv (1,3) \mod 4$
\item \label{i:com:3-2} \textbf{if} $(i,j)\equiv (3,2) \mod 4$:  we reach configuration $(j+1, i+1) \equiv (3,0) \mod 4$
\item \label{i:com:3-3} \textbf{if} $(i,j)\equiv (3,3) \mod 4$:  we reach configuration $(j+1, i+1) \equiv (0,0) \mod 4$
\end{enumerate}

In any case, rendezvous is achieved after at most three rounds.
\end{proof}

\begin{lemma}\label{lem: proof of algo, correct, opposite}
Assuming rigid movement, no crash, and robots having \textbf{opposite orientations} of the line joining them, then, Algorithm~\ref{algo:disoriented FSYNC} solves the SUIR problem.
\end{lemma}

\begin{proof}

Since the robots have opposite orientations, we know they either both see themselves on the right or they both both see themselves on the left. Of course the robots are not aware of this, but we saw in the previous remark that the opposite orientations are preserved during the whole execution (while robots are not gathered).

In this proof, $R\{i,j\}$ denotes a configuration where both robots see themselves on the right and one of them has level $i$ and the other as level $j$. Here, the order between $i$ and $j$ does not matter (hence the set notation). Similarly $L\{i,j\}$ denotes a configuration where both robots see themselves on the left, and one of them has level $i$ and the other as level $j$. 

Here, assuming without loss of generality that $i\leq j \mod 4$, we write $R\{i,j\}\equiv (k,l) \mod 4$, resp. $L\{i,j\}\equiv (k,l) \mod 4$, if and only if, $i\equiv k \mod 4$ and $j\equiv l \mod 4$.

To prove the lemma, we want to show that for any configuration $R\{i,j\}$ or $L\{i,j\}$, the robots achieve rendezvous.
Take an arbitrary configuration $(i,j) \in \Z^2$. We consider all 20 cases:

\begin{enumerate}
\item \label{i:opp:L0-0} \textbf{if} $L\{i,j\}\equiv (0,0) \mod 4$: rendezvous is achieved in one round.
\item \label{i:opp:L0-1} \textbf{if} $L\{i,j\}\equiv (0,1) \mod 4$: rendezvous is achieved in one round.
\item \label{i:opp:L0-2} \textbf{if} $L\{i,j\}\equiv (0,2) \mod 4$: rendezvous is achieved in one round.
\item \label{i:opp:L0-3} \textbf{if} $L\{i,j\}\equiv (0,3) \mod 4$: we reach configuration $R\{i{+}1, j{+}1\}{\equiv} (0,1) \mod 4$.

\item \label{i:opp:L1-1} \textbf{if} $L\{i,j\}\equiv (1,1) \mod 4$: rendezvous is achieved in one round.
\item \label{i:opp:L1-2} \textbf{if} $L\{i,j\}\equiv (1,2) \mod 4$: rendezvous is achieved in one round.
\item \label{i:opp:L1-3} \textbf{if} $L\{i,j\}\equiv (1,3) \mod 4$: we reach configuration $R\{i{+}1, j{+}1\}{\equiv} (0,2) \mod 4$.
\item \label{i:opp:L2-2} \textbf{if} $L\{i,j\}\equiv (2,2) \mod 4$: rendezvous is achieved in one round.
\item \label{i:opp:L2-3} \textbf{if} $L\{i,j\}\equiv (2,3) \mod 4$: we reach configuration $R\{i{+}1, j{+}1\}{\equiv} (0,3) \mod 4$.
\item \label{i:opp:L3-3} \textbf{if} $L\{i,j\}\equiv (3,3) \mod 4$: we reach configuration $R\{i, j\}{\equiv} (3,3) \mod 4$.

\item \label{i:opp:R0-0} \textbf{if} $R\{i,j\}\equiv (0,0) \mod 4$: rendezvous is achieved in one round.
\item \label{i:opp:R0-1} \textbf{if} $R\{i,j\}\equiv (0,1) \mod 4$: we reach configuration $L\{i{+}1, j{+}1\}{\equiv} (1,2) \mod 4$.
\item \label{i:opp:R0-2} \textbf{if} $R\{i,j\}\equiv (0,2) \mod 4$: rendezvous is achieved in one round.
\item \label{i:opp:R0-3} \textbf{if} $R\{i,j\}\equiv (0,3) \mod 4$: rendezvous is achieved in one round.
\item \label{i:opp:R1-1} \textbf{if} $R\{i,j\}\equiv (1,1) \mod 4$: we reach configuration $L\{i, j\}{\equiv} (1,1) \mod 4$.
\item \label{i:opp:R1-2} \textbf{if} $R\{i,j\}\equiv (1,2) \mod 4$: we reach configuration $L\{i{+}1, j{+}1\}{\equiv} (2,3) \mod 4$.
\item \label{i:opp:R1-3} \textbf{if} $R\{i,j\}\equiv (1,3) \mod 4$: we reach configuration $L\{i{+}1, j{+}1\}{\equiv} (0,2) \mod 4$.
\item \label{i:opp:R2-2} \textbf{if} $R\{i,j\}\equiv (2,2) \mod 4$: rendezvous is achieved in one round.
\item \label{i:opp:R2-3} \textbf{if} $R\{i,j\}\equiv (2,3) \mod 4$: rendezvous is achieved in one round.
\item \label{i:opp:R3-3} \textbf{if} $R\{i,j\}\equiv (3,3) \mod 4$: rendezvous is achieved in one round.
\end{enumerate}

In any case, rendezvous is achieved after at most three rounds.
\end{proof}

\begin{theorem}
Algorithm~\ref{algo:disoriented FSYNC} solves the SUIR problem with disoriented robots in FSYNC.
\end{theorem}

\begin{proof}
By Lemma~\ref{lem:distance decreases by fixed distance}, the distance between the two robots decreases by at least $\min(\delta, d/2)$ every two rounds. Hence, eventually, robots are at distance smaller than $\delta$ from one another and, from this point in time, movements are rigid. Assume now that movements are rigid.
If a robot crashes, then the rendezvous is achieved by using Lemma~\ref{lem: proof of algo, rigid, crash}. Otherwise, depending on whether the robots have a common orientation or opposite orientation of the line joining them (see Remark~\ref{rem:orientation common or opposite}), the Theorem follows either by using Lemma~\ref{lem: proof of algo, correct, common} or by using Lemma~\ref{lem: proof of algo, correct, opposite}.
\end{proof}

%\textbf{ST: il faudrait une idée du nombre de ronds pour le rendez-vous.}
%\textbf{QB: Max 3 round a partir du moment ou les mouvements sont rigides et si les crash arrives au début. Donc brutalement je dirait 5 round si les crash arrive n'importe quand. Et du coup $2d/\delta + 5$ (ou un truc du style) avec mouvements non-rigides}

\section{Concluding remarks}

We considered the problem of Stand Up Indulgent Rendezvous (SUIR). Unlike classical rendezvous, the SUIR problem is unsolvable in SSYNC even with the strongest assumptions: robots share a common $x-y$ coordinate system, and have access to infinite persistent memory. We demonstrate that it is nevertheless solvable in FSYNC without \emph{any} additional assumptions. 
%A natural question is whether other additional assumptions enable SUIR solvability in SSYNC. 
%One can observe that the classical probabilistic rendez-vous algorithm "when activated, with probability $\frac{1}{2}$, go to the other robot's position" also solves SUIR in a probabilistic sense in SSYNC (if no robot crashes, the same arguments as for rendez-vous applies; if a robot crashes, the other robot joins its position in two steps in expectation). 
%\input{open}
%It remains unknown whether luminous robots enable deterministic SUIR solvability in SSYNC.  
A natural open question is related to the optimality (in time) of our algorithm. 

Also, we would like to investigate further the possibility of deterministic strong gathering for $n\geq 3$ robots. It is known that executing a single robot at a time in SSYNC is insufficient~\cite{DefagoGMP06,DefagoPCMP16}, but additional hypotheses may make the problem solvable.

%In the FSYNC setting, $n\geq3$ robots gathering can be done without any assumption on the starting configuration~\cite{CohenP05}, even without multiplicity detection~\cite{balabonski19tcs}. It would be interesting to generalize SUIR to an arbitrary number of robots $n$ (so, this problem could be Stand Up Resilient Gathering, SUIG). Current algorithms for gathering~\cite{CohenP05} do not solve SUIG as they use the center of gravity of robots~\cite{CohenP05} or occupied positions~\cite{balabonski19tcs}, hence the correct robots only approach the position of the crashed robot, but may never reach it. 
%Similarly, crash tolerant gathering protocols~\cite{bouzid13icdcs,bramas15wait} only guarantee the gathering of correct process on the same location, irrespective of the location of a crashed process. We plan to study the feasibility of SUIG.

%Open question: optimality in term of number of rounds?

\bibliographystyle{splncs04}
\bibliography{biblio}

\newpage
%\appendix
%\section{Omitted Proofs}

%\contentFromMainDocument

\end{document}